\documentclass[11pt,fleqn]{article}
\usepackage[round]{natbib}
\usepackage{amsmath,amssymb,amsthm}
\usepackage{xcolor}
\usepackage[pdfstartview=FitH,bookmarksopen=false,
colorlinks=true,linkcolor=blue,citecolor=blue]{hyperref}
\usepackage[left=2cm,right=2cm,top=2cm,bottom=2cm]{geometry}
\usepackage{graphicx}
\usepackage{caption}
\newtheorem{theorem}{Theorem}

\newtheorem{observation}{Observation}

\newtheorem{remark}{Remark}

\newcounter{mycounter}
\begin{document}

        \title{\bf Highly mutually dependent unions and new axiomatizations of the Owen value}
	\author{Songtao He,\, Erfang Shan\thanks{{\em Corresponding author}. E-mail addresses: efshan@shu.edu.cn (E. Shan), hesongtao@shu.edu.cn (S. He), zhq5036@163.com (H. Zhou)}, \, Hanqi Zhou
	}
	\date{}
	\maketitle \baselineskip 17 pt
\begin{abstract}
	
 The Owen value	is an well-known allocation rule for cooperative games with coalition structure.
 In this paper, we introduce the concept of highly mutually dependent unions. Two unions in a cooperative game with coalition structure are said to be highly mutually dependent if  any pair of players, with one from each of the two unions, are mutually dependent in the game.
 Based on this concept, we introduce two axioms: weak mutually
dependent between unions and differential marginality of inter-mutually
dependent unions. Furthermore, we also propose another two axioms: super inter-unions
marginality and  invariance across games, where the former one is based on the concept of the inter-unions marginal contribution. By using the axioms and combining with some standard axioms, we present three axiomatic characterizations of the Owen value.
	
	\bigskip
	\noindent {\bf Keywords}: TU-game; coalition structure; CS-game; Owen value; weak mutually dependent between unions; invariance across games
	
	\medskip
	\noindent {\bf AMS (2000) subject classification:} 91A12
	
	\noindent {\bf JEL classification:} C71
\end{abstract}

\section{Introduction}\label{section1}

In many negotiations, due to some common interests, some players prefer to cooperate together
than with others. The basic tool for analyzing such negotiations  is the cooperative games with coalition structure (CS-games) \citep{1974_aumann_dreze},
in which a coalition structure is defined as a partition of the player set into disjoint unions.
\cite{1977_owen} generalized the famous Shapley value \citep{1953_shapley} to the setting of CS-games and  characterized
it by means of five axioms: efficiency, additivity, symmetry within unions, symmetry between unions and the null player property.
Later on, numerous axiomatic characterizations of the Owen value have been developed (see, e.g., \cite{wi1992,vaz1997,ha-1999,2001_hamiache,2010_ca,Lo2016,2007_kh,2008_albizuri,2021_hu,Cas2023,chen2024,2024_he}).

The axiom of symmetry between unions extends the standard axiom  of symmetry  \citep{1953_shapley} to CS-games. It states that two unions should receive the same payoff if they are symmetric in the quotient game played by unions.
In real life, it is possible that unions are mutually dependent or that any pair of players from two different unions are mutually dependent.
For instance,  in a certain industry, there are two different business unions A and B. Union A is composed of some large manufacturers, and union B is made up of some powerful sales channel providers. These two unions are mutually dependent because union A needs union B to promote the products it manufactures to the market, while union B requires union A to supply high-quality products for sale. Generally speaking, the manufacturers from union A and the sellers from union B are also interdependent with each other.
To describe this situation, we introduce the notion of highly mutually dependent unions. Based on this definition, we propose two  axioms and establish three new axiomatizations of the Owen value by combining them with several standard axioms.

Two unions in CS-games are said to be highly mutually dependent if any pair of players, with one player from one union and the other player from the other union respectively, are mutually dependent.
Based on this notion, we first propose an property, called weak mutually dependent between unions, which is a relaxation of the standard axiom of symmetry between unions. This axiom requires that if two unions  in a game are highly mutually dependent, then the sum of all players' payoff within these two unions should be equal.
We also introduce  three  properties involving marginality and invariance: differential marginality of inter-mutually dependent unions, super inter-unions marginality and invariance across games.
The axiom of differential marginality of inter-mutually dependent unions states that if two unions are highly
mutually dependent in the differential game for two games, then the differences of the total payoff within
these two unions should be equal in the two games.
The axioms of super inter-unions marginality states that if the inter-unions marginal contributions of a player\footnote{For a CS-game $(N,v,\mathcal{B}) \in \mathcal{CG}$ and $i \in B_p$ with $B_p \in \mathcal{B}$,
 $v\big(\big(\bigcup_{B_i\in \mathcal{B}'}\cup S\big)\cup \{i\}\big)-v\big(\bigcup_{B_i\in \mathcal{B}'}\cup S\big)$ is called the {\em inter-unions marginal contribution} of player $i$ to coalition $\bigcup_{B_i\in \mathcal{B}'}\cup S$ where $S\subseteq B_p\setminus \{i\}$ and  $\mathcal{B}'\subseteq \mathcal{B}\setminus \{B_p\}$.}  in two games are the same, then the payoff he receives should be the same.
The axiom of invariance across games
 describes the situation in which all unions have the same productive identity (simultaneously productive or unproductive) in two different games.
 It requires that for any two union-wise mutually dependent CS-games (see subsection \ref{subsec1}), provided all members within a union have the same productive identity  and all non-null players within this union are mutually dependent in the two games, the payoffs of all players within this union should remain unchanged across the two games.
 The axioms of marginality and invariance  are often used in the characterization of values for cooperative games (see, e.g., \cite{1985_young,2007_kh,2010_ca,2020_ma,2023_shan,2024_shan,2024_chenj,2024_cheo,2024_he}).

The purpose of this paper is to propose new axiomatic characterizations of the Owen value by making use of the axioms mentioned above.
Firstly, by combining the axiom of weak mutually dependent between unions with the standard axioms: efficiency, additivity, symmetry within unions, and the null player property, we provide the first axiomatic characterization of the Owen value.
Secondly, we show that the axiom of super inter-unions marginality,
together with the axioms of weak mutually dependent between unions, efficiency and symmetry within unions, uniquely characterizes the Owen value.
Finally, we show that the Owen value is characterized by the axioms of differential marginality of inter-mutually dependent unions, efficiency, null player out and differential marginality of mutually dependent players within unions.

The rest of this paper is organized as follows.  Section \ref{section2} contains some preliminaries
for TU-games, CS-games, and the Owen value. Section \ref{section3}  states our axioms and characterizes the Owen value. Finally, in Section \ref{section4}, we conclude this paper.

\section{Preliminaries}\label{section2}
\subsection{TU-games}\label{section2.1}

Denote $\mathcal{U}$ as an infinite set of potential players and let $\mathcal{N}$ be the set of all non-empty and finite subsets of $\mathcal{U}$. For any $N \in \mathcal{N}$, a \textit{cooperative game with transferable utility} (\textit{TU-game}) is a pair $(N,v)$, where  $v :2^{N} \to \mathbb{R}$ is the \textit{characteristic function} satisfying  $v(\emptyset)=0$. A subset $S \subseteq N$ is called a \textit{coalition}
 and $v(S)$ is the \textit{worth} of coalition $S$ in the cooperation, the set $N$ is called the \textit{grand coalition}.  Let ${\mathcal{G}^{N}}$ be the set of all TU-games over $N$, and let $\mathcal{G} := \bigcup_{N \in \mathcal{N}}\mathcal{G}^{N}$.  We denote $(N,v)$ as $v$ if there is no ambiguity.  $|S|$ denotes the cardinality of a set $S$.

The \textit{null game} ${\bf 0} \in {\mathcal{G}^{N}}$ is defined as ${\bf {0}}(S)=0$ for all $S \subseteq N$. The \textit {unanimity game} $(N,u^N_T)$ is given by $u^N_T(S)=1$ for all $T \subseteq S$ and $u^N_T(S)=0$ otherwise. The family $\{u_T:\emptyset \neq T \subseteq N\}$ forms a basis on ${\mathcal{G}}^{N}$, thus any TU-game $v \in {\mathcal{G}}^{N}$ can be uniquely represented by
\begin{eqnarray*}
	v=\sum_{T\subseteq N, T \neq \emptyset} \lambda_T(v)u_{T}^{N},
\end{eqnarray*}
where $\lambda_T(v) = \sum_{S \subseteq T} {(-1)}^{|T|-|S|} v(S)$ is called the \textit{Harsanyi dividend} of the coalition $T$ \citep{1959_harsanyi}. The worth of every non-empty coalition $S \subseteq N$ can be written as
\begin{equation}\label{equation1}
	v(S) = \sum_{\emptyset \neq T \subseteq S} \lambda_T(v).
\end{equation}

For a TU-game $v \in {\mathcal{G}}^{N}$, a player $i \in N$ is called a \textit{null player} in $v$ if $v(S\cup\{i\})=v(S)$ for all $S \subseteq {N \setminus   \{i\}}$. Player $i$ is  a \textit{necessary player} in $v$ if $v(S)=0$ for any $S \subseteq N \setminus   \{i\}$. Two players $i,j\in N$ are \textit{symmetric} in $v$ if  $v(S\cup\{i\})-v(S)=v(S\cup\{j\})-v(S)$ for all $S \subseteq N \setminus   \{i,j\}$ and they are  \textit{mutually dependent} in $v$ if $v(S\cup\{i\})-v(S)=v(S\cup\{j\})-v(S)=0$ for any $S \subseteq N \setminus   \{i,j\}$ \citep{1995_nowak}. Obviously,  any two necessary players are mutually dependent, and any two mutually dependent players are symmetric. For two TU-games $(N,v), (N,w)\in \mathcal{G}^N$, a player $i \in N$ has the same \textit{productive identity} in these two games if it is either a null player in both games $v$ and $w$ or a non-null player in the both games.

Based on the definitions above, we have the following basic observation \citep{2024_shan}.

\begin{observation}\label{ob1}

    (i)  A player $i\in N$ is a null player in $v$ if and only if $\lambda_T(v)=0$ for all $T\subseteq N$ containing  $i$;

    (ii) For any $T \subseteq N$, if there is a necessary player $i$ in $v$ such that $i \notin T$, then $\lambda_T(v)=0$;

    (iii) Two players $i$ and $j$ are mutually dependent in $v$ if and only if  $\lambda_{T\cup\{i\}}(v)=\lambda_{T\cup\{j\}}(v)=0$ for all $T \subseteq N \setminus   \{i,j\}$.

\end{observation}

An \textit{allocation rule} or a \textit{value} $\varphi$ over ${\mathcal{G}}^{N}$ is a function that assigns a payoff vector $\varphi(N,v) \in \mathbb{R}^{N}$ to any TU-game $(N,v) \in {\mathcal{G}}^{N}$.

One of the famous values in TU-games is the \textit{Shapley value} \citep{1953_shapley}, which is given by the form based on the Harsanyi dividend as
\begin{equation}\label{equation2}
	Sh_i(N,v)=\sum_{\emptyset\neq T \subseteq N :i \in T} \frac{\lambda_T(v)}{|T|}, \ \ \mbox{for all $i \in N$}.
\end{equation}

\subsection{CS-games and the Owen value}\label{section2.2}

Let $N \in \mathcal{N}$, a \textit{coalition structure} $\mathcal{B}=\{B_1,B_2,...,B_m\}$ is a \textit{partition} over the players set $N$, in which each $B_p \in \mathcal{B}$ is called a \textit{union}, i.e., $\bigcup_{p \in M}{B_p}=N$ and $B_p \cap B_q=\emptyset$ when $p \neq q$, where $M=\{1,2,...,m\}$ is the set of indices of unions.  $\mathcal{B}^n=\{\{i\}:i \in N\}$ where every union is a singleton and $\mathcal{B}^N=\{N\}$ where the grand coalition forms  are two trivial coalition structures. The set of all coalition structures on $N$ is denoted by $\mathcal{B}(N)$.  For a given coalition $T \subseteq N$ and a coalition structure $\mathcal{B}=\{B_1,B_2,...,B_m\}$, we denote the coalition structure restricted to $T$ as $\mathcal{B}_{|T} = \{T_p : T_p = B_p \cap T \neq  \emptyset \text{ for all } p \in M\}$. We often denote $B(i)$ as the union that contains the player $i$.

A \textit{TU-game with a coalition structure} (\textit{CS-game}) is a triple $(N,v,\mathcal{B})$ where $(N,v)\in {{\mathcal{G}}^{N}}$ is a TU-game and $\mathcal{B} \in \mathcal{B}(N)$ is a coalition structure on $N$. Denote by ${\mathcal{CG}}^{N}$ the set of all CS-games over $N$  and the set of all the CS-games by $\mathcal{CG}:= \bigcup_{N \in \mathcal{N}}{\mathcal{CG}^{N}}$.

For a given CS-game $(N,v,\mathcal{B})$ where $\mathcal{B}=\{B_1,B_2,...,B_m\}$, the \textit{quotient game} associated with $(N,v,\mathcal{B})$ is denoted as $(M,{v}^{\mathcal{B}})$, in which ${v}^{\mathcal{B}}(R)=v(\bigcup_{r\in R}{B_r})$ for all $R \subseteq M$.

A union $B_p \in \mathcal{B}$ is  a \textit{null  union} in $v$ if $p$ is  null  in $(M,{v}^{\mathcal{B}})$, and  $B_p \in \mathcal{B}$ is  a \textit{necessary  union} in $v$ if $p$ is  necessary in $(M,{v}^{\mathcal{B}})$. Two unions, $B_p, B_q\in \mathcal{B}$,  are  \textit{mutually dependent (symmetric) unions} in $v$ if $p$ and $q$ are mutually dependent (symmetric) in $(M,{v}^{\mathcal{B}})$; they are  \textit{highly mutually dependent unions} in $v$ if any player $i \in B_p$ and any player $j \in B_q$ are mutually dependent in
$v$.
 By Observation \ref{ob1} (iii) and the formula $\lambda_R(v^{\mathcal{B}})=\sum_{T\in S^R}\lambda_T(v)$ for any $R\subseteq M$ where $S^R=\{T\subseteq N:\ \ T\cap B_r\neq \emptyset \ \mbox{for all $r\in R$}\}$ \citep{2025_he}, we can deduce that {\bf  if two unions are highly mutually dependent in $v$, then they must be mutually dependent}.

For two games $(N,v,\mathcal{B}),(N,w,\mathcal{B}) \in {\mathcal{CG}}^{N}$, a union $B \in \mathcal{B}$ is said to have the same  productive identity in both $v$ and $w$ if  $B$ is either a null union in both games $v$ and $w$ or a non-null union in the two games.

A \textit{coalitional value} on $\mathcal{CG}^N$ is a function $\varphi$ that assigns a payoff vector $\varphi(N,v,\mathcal{B}) \in
\mathbb{R}^N$ for all CS-game $(N,v,\mathcal{B}) \in \mathcal{CG}^N$.

The \textit{Owen value} \citep{1977_owen} is a  well-known coalitional value on $\mathcal{CG}^N$, given by
\begin{equation}\label{equation3}
	Ow_i(N,v,\mathcal{B}) = \sum_{\emptyset \neq T \subseteq N : T \ni i} \frac{\lambda_T(v)}{|B(i) \cap T| m_T}, \ \ \mbox{for all $i \in N$}
\end{equation}
where $m_T=|\mathcal{B}_{|T}|$. It can also be written as
\begin{equation}\label{equation4}
	\begin{aligned}
		Ow_i(N,v,\mathcal{B}) &= \sum_{R \subseteq M \setminus   \{h\} }\sum_{S \subseteq
B_h \setminus   \{i\}}\frac{|R|!(|M|-|R|-1)!}{|M|!}\frac{|S|!(|B_h|-|S|-1)!}{|B_h|!}\\
		&\quad\times [v(Q(R)\cup S\cup\{i\})-v(Q(R)\cup S)],
	\end{aligned}
\end{equation}
for all $i \in B_h$ and all $B_h \in \mathcal{B}$, where $Q(R)=\bigcup_{r \in R}B_r$.
This difference $v(Q(R)\cup S\cup\{i\})-v(Q(R)\cup S)$ is called the {\em inter-unions marginal contribution} of player $i$ to coalition $Q(R)\cup S$.

The following observation was verified by \cite{wi1992}.
\begin{equation}\label{equation5}
	\sum_{i \in B_p}Ow_i(N,v,\mathcal{B}) = Sh_p(M,v^{\mathcal{B}})\ \  \mbox{for any   $B_p \in \mathcal{B}$}.
\end{equation}

Below, we introduce the axioms  used in the characterizations of the Owen value.

{\em Efficiency} ({\bf E}). For any $(N, v, \mathcal{B}) \in \mathcal{CG}^N$, $\sum_{i\in N}\varphi_i(N, v, \mathcal{B})=v(N)$.

{\em Additivity} ({\bf A}).
For any $(N, v, \mathcal{B}), (N, w, \mathcal{B}) \in \mathcal{CG}^N$, $\varphi(N,v,\mathcal{B}) + \varphi(N,w,\mathcal{B}) = \varphi(N,v+w,\mathcal{B})$.

{\em Null player property} ({\bf N}). For any $(N, v, \mathcal{B}) \in \mathcal{CG}^N$,  $\varphi_i(N, v, \mathcal{B})=0$ if $i\in N$ is a null player in $v$.

{\em Null player out} ({\bf NPO}). For any $(N, v, \mathcal{B}) \in \mathcal{CG}$, if $i\in N$ is a null player in $v$, then $\varphi_j(N, v, \mathcal{B})=\varphi_j(N \setminus   \{i\}, v, \mathcal{B}_{|N \setminus   \{i\}})$ for all $j \in N \setminus   \{i\}$.

{\em Symmetry} ({\bf S}). For any $(N, v, \mathcal{B}) \in \mathcal{CG}^N$,  $\varphi_i (N,v,\mathcal{B})=\varphi_j (N,v,\mathcal{B})$ if two players $i,j\in N$ are symmetric in $v$.

{\em Symmetry within unions} ({\bf SWU}). For any $(N, v, \mathcal{B}) \in \mathcal{CG}^N$ and any two players $i,j \in B \in \mathcal{B}$,  $\varphi_i (N,v,\mathcal{B})=\varphi_j (N,v,\mathcal{B})$ if $i$ and $j$ are symmetric in $v$.

{\em Symmetry between unions} ({\bf SBU}). For any $(N, v, \mathcal{B}) \in \mathcal{CG}^N$ and any two unions $B_p, B_q \in \mathcal{B}$,  $\sum_{i \in B_p}\varphi_i (N,v,\mathcal{B})=\sum_{i \in B_q}\varphi_i (N,v,\mathcal{B})$ if $B_p$ and $B_q$ are symmetric in $v$.

{\em Marginality} ({\bf M}). For any $(N, v, \mathcal{B})$, $(N, w, \mathcal{B}) \in \mathcal{CG}^N$, if $v(S\cup\{i\})-v(S)=w(S\cup\{i\})-w(S)$ for all $S\subseteq N\setminus \{i\}$, then $\varphi_i(N, v, \mathcal{B}) = \varphi_i(N, w, \mathcal{B})$.


{\em Differential marginality of mutually dependent players within unions} ({\bf UDM$_{md}$}). For any $(N, v, B)$, $(N, w, B) \in \mathcal{CG}^N$, if $i,j \in B_k \in \mathcal{B}$ are mutually dependent players in $v-w$, then $\varphi_i(N, v, B) - \varphi_j(N, v, B) = \varphi_i(N, w, B) - \varphi_j(N, w, B)$.

This axiom requires that if two players in a union are  mutually dependent in the differential game $v-w$, then the differences of the payoffs for the two players in the both games are the same.

{\em Differential marginality between mutually dependent unions} ({\bf DMU$_{md}$}). For any $(N, v, B)$, $(N, w, B) \in \mathcal{CG}^N$, if $B_p, B_q \in \mathcal{B}$ are mutually dependent unions in $v-w$, then $\sum_{i \in B_p}\varphi_i(N, v, B) - \sum_{i \in B_p}\varphi_i(N, w, B) = \sum_{i \in B_q}\varphi_i(N, v, B)-\sum_{i \in B_q}\varphi_i(N, w, B)$.

This axiom states that if two unions are  mutually dependent in the differential game $v-w$, then 
the difference between the sums of payoffs of two unions in the both games should be equal.

\cite{1977_owen} characterized the Owen value by  efficiency ({\bf E}), additivity ({\bf A}), symmetry within unions ({\bf SWU}), symmetry between unions ({\bf SBU}), the null player property ({\bf N}). \cite{2007_kh} gave an alternative characterization of the Owen value by means of efficiency ({\bf E}), symmetry within unions ({\bf SWU}), symmetry between unions ({\bf SBU}) and marginality ({\bf M}).
Recently, \cite{2024_he} characterized the Owen value by efficiency ({\bf E}), differential marginality of mutually dependent players within unions ({\bf UDM$_{md}$}), differential marginality between mutually dependent unions ({\bf DMU$_{md}$}) and null player out ({\bf NPO}).

\section{Axiomatizations of the Owen value}\label{section3}
In this section we will provide  three characterizations of the Owen value. For this purpose, we will introduce more notions and axioms for CS-games in the following subsection.

\subsection{Some new axioms}\label{subsec1}
Based on the concepts of highly mutually dependent unions and inter-unions marginal contribution introduced in the previous section, we first present
the following three axioms.

{\em Weak mutually dependent between unions} ({\bf MBU$^-$}). For any $(N, v, \mathcal{B}) \in \mathcal{CG}$ and any two unions $B_p, B_q \in \mathcal{B}$, if   $B_p$ and  $B_q$ are highly mutually dependent in $v$, then $\sum_{i \in B_p}\varphi_i (N,v,\mathcal{B})=\sum_{i \in B_q}\varphi_i (N,v,\mathcal{B})$.

This axiom states that if  two unions are highly mutually dependent in a game, then the sum of all players' payoff within these two unions should be equal.

{\em Differential marginality of inter-mutually dependent unions} ({\bf DMU$_{md}^-$}). For any games $(N, v, \mathcal{B})$, $(N, w, \mathcal{B}) \in \mathcal{CG}$ and any two unions $B_p,B_q \in \mathcal{B}$, if   $[v(S\cup\{i\})-v(S)]=[w(S\cup\{i\})-w(S)]$ and $[v(S\cup\{j\})-v(S)]=[w(S\cup\{j\})-w(S)]$ for all $i\in B_p$, all $j\in B_q$ and all $S \subseteq N \setminus   \{i,j\}$,  then $\sum_{i \in B_p}\varphi_i(N, v, B) - \sum_{i \in B_p}\varphi_i(N, w, B) = \sum_{i \in B_q}\varphi_i(N, v, B)-\sum_{i \in B_q}\varphi_i(N, w, B)$.

Note that the hypothesis of the axiom is satisfied if and only if  the unions $B_p$ and $B_p$ are highly mutually dependent in the game $v-w$.
This axiom states that if two unions are highly mutually dependent in the differential game $v-w$, then the differences of the total payoff within these two unions should be equal in both $v$ and $w$.

{\em Super inter-unions marginality} ({\bf IUM$^+$}). For any games $(N,v,\mathcal{B})$, $(N,w,\mathcal{B})$ $\in \mathcal{CG}$ and $i \in B \in \mathcal{B}_p$, if $v(Q(R) \cup S\cup\{i\}) - v(Q(R) \cup S) = w(Q(R) \cup S\cup\{i\}) - w(Q(R) \cup S)$ for all $R \subseteq M \setminus   \{p\} $ and all $S \subseteq B_p \setminus   \{i\}$, then $\varphi_i(N,v,\mathcal{B})=\varphi_i(N,w,\mathcal{B})$.

This axiom states that, for any two CS-games, if a player has the same  inter-unions marginal contribution in both two games, then this player should receive equal payoff in both two games.

For any $(N, v, \mathcal{B}), (N, w, \mathcal{B})\in \mathcal{CG}^N$ such that $v(N) = w(N)$, they are called {\em union-wise mutually dependent} if all unions have the same productive identity (every union is simultaneously either  null  or non-null) in both $v$ and $w$, and all non-null unions are mutually dependent in both two games. Based on this notion, we present the following property.

\textit{Invariance across games} \textbf{(IAG)}. For any pair of union-wise mutually dependent games $(N, v, \mathcal{B})$ and $(N, w, \mathcal{B}) \in {\mathcal{CG}}^N$ and any $B_l \in \mathcal{B}$, if all players in $B_l$ have the same productive identity and all non-null players are mutually dependent in both $v$ and $w$, then $\varphi_i(N, v, \mathcal{B}) = \varphi_i(N, w, \mathcal{B})$ for all $i \in B_l$.

This axiom states that for any two union-wise mutually dependent CS-games, if all members within a union have the same productive identity and all non-null players within this union are mutually dependent in both two games, then the payoffs of all players in this union should remain invariant across the two games.

 As we have seen, the highly mutually dependent unions are also mutually dependent in $v$. This implies that
 the axiom of weak mutually dependent between unions ({\bf MBU$^-$}) is a weakened version of the classical axiom of symmetry between unions ({\bf SBU}) \citep{1977_owen}, and the axiom of differential marginality of inter-mutually dependent unions ({\bf DMU$_{md}^-$}) is weaker than the axiom of differential marginality between mutually dependent unions ({\bf DMU$_{md}$}) \citep{2024_he}.
 Furthermore, we observe that the axiom of super inter-unions marginality ({\bf IUM$^+$}) is slightly stronger than  the classical marginality axiom ({\bf M}) \citep{1985_young}.

\subsection{Axiomatizations}
We give the first characterization of the Owen value by replacing the axiom of symmetry between unions ({\bf SBU}) in the Owen's axiomatization \citep{1977_owen} with the axioms of weak mutually dependent between unions ({\bf MBU$^-$})  and invariance across games ({\bf IAG}).

\begin{theorem}\label{thm1}
	The $Owen$ value is the unique value on $\mathcal{CG}^N$ that satisfies efficiency {\rm ({\bf E})},  additivity {\rm({\bf A})}, symmetry within unions {\rm({\bf SWU})}, weak mutually dependent between unions {\rm({\bf MBU$^-$})}, invariance across games {\rm({\bf IAG})} and the null player property {\rm({\bf N})}.
\end{theorem}

\begin{proof}
{\em Existence}.
It is well known that the Owen value satisfies {\bf E}, {\bf A},  {\bf N}, {\bf SWU} and {\bf SBU} \citep{1977_owen}.  Note that {\bf SBU} implies {\bf MBU$^-$}. Thus the Owen value satisfies {\bf MBU$^-$}. Regarding {\bf IAG}, consider
$(N, v, \mathcal{B}), (N, w, \mathcal{B}) \in \mathcal{CG}^N$ and $B_p \in \mathcal{B}$, which satisfies the assumption for {\bf IAG}.

We first claim that $\sum_{i \in B_p}Ow_i(N,v,\mathcal{B})=\sum_{i \in B_p}Ow_i(N,w,\mathcal{B})$.
Indeed, if $B_p$ is a null union in both $v$ and $w$, then, by  (\ref{equation5}), we obtain
\begin{equation}\label{equation6}
	\sum_{i \in B_p}Ow_i(N,v,\mathcal{B})\stackrel{(\ref{equation5})}{=}Sh_p(M,v^{{\mathcal{B}}})=0
=Sh_p(M,w^{{\mathcal{B}}})\stackrel{(\ref{equation5})}{=}\sum_{i \in B_p}Ow_i(N,w,\mathcal{B}),
\end{equation}
where the second and third equalities are derived from the fact that the Shapley value satisfies {\bf N}.
Now we may assume that $B_p$ is a non-null union in both $v$ and $w$. Let $M'=\{l \in M \mid B_l \text{ is non-null union in } v\}$. Since all the unions $B \in \mathcal{B}$ have the same productive identity and all non-null unions are mutually dependent in $v$ and $w$, then,
by {\bf E} and {\bf S} of the Shapley value \citep{1953_shapley}, we have
\begin{eqnarray*}
v(N)=v^{\mathcal{B}}(M)=\sum_{l\in M}Sh_l(M,v^{\mathcal{B}})\stackrel{(\ref{equation6})}{=}\sum_{l\in M'}Sh_l(M,v^{\mathcal{B}})=|M'|Sh_p(M,v^{\mathcal{B}}).
\end{eqnarray*}
Similarly, we have
\begin{eqnarray*}
w(N)=|M'|Sh_p(M,w^{\mathcal{B}}).
\end{eqnarray*}
Since $v(N)=w(N)$, by (\ref{equation5}), we have
\begin{equation*}\label{equation7}
	\sum_{i \in B_p}Ow_i(N,v,\mathcal{B})\stackrel{(\ref{equation5})}{=}Sh_p(M,v^{{\mathcal{B}}})
=Sh_p(M,w^{{\mathcal{B}}})\stackrel{(\ref{equation5})}{=}\sum_{i \in B_p}Ow_i(N,w,\mathcal{B}).
\end{equation*}

Now we show that $Ow_i(N, v, \mathcal{B})=Ow_i(N, w, \mathcal{B})$ for all $i\in B_p$. Let $NB_p(v)$ and $NB_p(w)$ be the set of all the null players of $B_p$ in $v$ and $w$, respectively. By the assumption that $NB_p(v)=NB_p(w)$ and {\bf N} of the Owen value,
If $i\in NB_p(v)=NB_p(w)$,  then $Ow_i(N, v, \mathcal{B})=0=Ow_i(N, w, \mathcal{B})$. For $i\in B_p\setminus NB_p(v)$, since the players in $B_p\setminus NB_p(v)$ are mutually dependent in $v$ and $w$,
 by {\bf SWU},
\begin{eqnarray*}
|B_p\setminus NB_p(v)|Ow_i(N,v,\mathcal{B})=\sum_{i \in B_p\setminus NB_p(v)}Ow_i(N,v,\mathcal{B})=\sum_{i \in B_p}Ow_i(N,v,\mathcal{B}).
\end{eqnarray*}
Similarly, we have
\begin{eqnarray*}
|B_p\setminus NB_p(w)|Ow_i(N,w,\mathcal{B})=\sum_{i \in B_p}Ow_i(N,w,\mathcal{B}).
\end{eqnarray*}
This implies that  $Ow_i(N, v, \mathcal{B})=Ow_i(N, w, \mathcal{B})$ for all $i\in B_p\setminus NB_p(v)$ ($=B_p\setminus NB_p(w)$).
Therefore, the Owen value satisfies {\bf IAG}.

{\em Uniqueness}.
    Let $\varphi$ be a value on $\mathcal{CG}^N$ that satisfies the above axioms. Let us show that
    $\varphi=Ow$. For any $(N, v, \mathcal{B})\in \mathcal{CG}^N$, we know that $v=\sum_{T\subseteq N, T \neq \emptyset} \lambda_T(v)u_{T}^{N}$. By {\bf A}, $\varphi(N,v, \mathcal{B})=\sum_{T\subseteq N, T \neq \emptyset}\varphi(N,\lambda_T(v)u_{T}^{N},\mathcal{B})$.
    Thus, it is sufficient to show that $\varphi(N,\alpha u_{T}^{N},\mathcal{B})=Ow(N,\alpha u_{T}^{N},\mathcal{B})$ for all $ T\subseteq N$ with $T\neq \emptyset$ where $\alpha$ is an arbitrary real constant. By {\bf N}, we only need to consider the cases when $\alpha u_{T}^{N}\neq 0$.

    For any  player $i \in N \setminus   T$, note that every $i\in N \setminus   T$ is a null player in $\alpha u_{T}^{N}$. Then, by {\bf N},
    \begin{equation}\label{equation9}
    	\varphi_i(N,\alpha u_{T}^{N},\mathcal{B})=Ow_i(N,\alpha u_{T}^{N},\mathcal{B})=0.
    \end{equation}

    It remains to show that $\varphi_i(N,\alpha u_{T}^{N},\mathcal{B})=Ow_i(N,\alpha u_{T}^{N},\mathcal{B})$ for all $i\in T$.
    Let $i\in T$ and let
    $$w_1=\alpha u^N_{T_1},\ \ \ w_2=\alpha u^N_{T_2},$$
    where
    \[T_1 = T \cup \bigg(\bigcup_{B \in \mathcal{B} \setminus \{B(i)\}: \ B \cap T \neq \emptyset} B\bigg), \ \ T_2=T_1 \cup B(i).\]
    Note that the players of $T$, $T_1$ and $T_2$ are all possible non-null players in $\alpha u_{T}^{N}$, $w_1$ and $w_2$, respectively.
By Observation \ref{ob1}, it is easily verified that $(N,\alpha u_{T}^{N},\mathcal{B})$ and $(N,w_1,\mathcal{B})$ are union-wise mutually dependent, and  all players of $B(i)$ have the same productive identity in both $\alpha u_{T}^{N}$ and $w_1$, and all the non-null players in $B(i)$, i.e., the players in $B(i)\cap T$ ($=B(i)\cap T_1$) are mutually dependent in both $\alpha u_{T}^{N}$ and $w_1$. Similarly, one sees that   $(N,w_1,\mathcal{B})$ and $(N,w_2,\mathcal{B})$ are union-wise mutually dependent, and for any $B \in \mathcal{B} \setminus \{B(i)\}$, all players of $B$ have the same productive identity in both $w_1$ and $w_2$, and
all the non-null players of $B$, i.e., the players in $B\cap T_1$ ($=B\cap T_2$)  are mutually dependent in both $w_1$ and $w_2$.
Then, by {\bf IAG}, we have
\begin{eqnarray}
&&\varphi_j(N,\alpha u_{T}^{N},\mathcal{B})=\varphi_j(N,w_1,\mathcal{B})\ \ \ \mbox{for all $i \in B(i)$}. \label{equation10}\\
&&\varphi_j(N,w_1,\mathcal{B})=\varphi_j(N,w_2,\mathcal{B})\ \ \ \mbox{for all $j \in B \in \mathcal{B} \setminus \{B(i)\}$}. \label{equation11}
\end{eqnarray}

We claim that  $\varphi_j(N,w_2,\mathcal{B})=Ow_j(N,w_2,\mathcal{B})$ for all $j\in N$.
For any $j\in N \setminus   T_2$, note that $j$ is a null player in $w_2$. Then, by {\bf N},
$\varphi_j(N,w_2,\mathcal{B})=Ow_j(N,w_2,\mathcal{B})=0$.
For any $j\in T_2$,
note that all players in $T_2$ are mutually dependent in $w_2$, and for any $B \in \mathcal{B}$ with $B \cap T_2 \neq \emptyset$,
$B \subseteq T_2$. By {\bf E}, {\bf N}, {\bf SWU},  and {\bf MBU$^-$}, we have
\begin{eqnarray*}
 w_2(N)&\stackrel{\bf E}{=}& \sum_{j \in N}\varphi_j(N,w_2,\mathcal{B})
    	\stackrel{\bf N}{=}\sum_{j \in T_2}\varphi_j(N,w_2,\mathcal{B}) \nonumber \\
    	&\stackrel{\bf SWU,\, MBU^-}{=}&
    	|\mathcal{B}_{|T_2}||B(j)|\varphi_j(N,w_2,\mathcal{B}).
    \end{eqnarray*}
 Similarly, we have $w_2(N)=|\mathcal{B}_{|T_2}||B(j)|Ow_j(N,w_2,\mathcal{B})$. This implies that $\varphi_j(N,w_2,\mathcal{B})=Ow_j(N,w_2,\mathcal{B})$ for all $j\in T_2$.
 Consequently, $\varphi_j(N,w_2,\mathcal{B})=Ow_j(N,w_2,\mathcal{B})$ for all $j\in N$.

Now, combining (\ref{equation10}), (\ref{equation11}) and {\bf E}, we obtain
    \begin{eqnarray*}
    	\sum_{j \in B(i)}\varphi_j(N,\alpha u_{T}^{N},\mathcal{B})&\stackrel{(\ref{equation10})}{=}&\sum_{j \in B(i)}\varphi_j(N,w_1,\mathcal{B}) \nonumber \\
    	&\stackrel{{\bf E }}{=}&
    	w_1(N)-\sum_{j \in \mathcal{B} \setminus   \{B(i)\}}\varphi_j(N,w_1,\mathcal{B}) \nonumber \\ &\stackrel{(\ref{equation1}),\, (\ref{equation11})}{=}&w_2(N)-\sum_{j \in \mathcal{B} \setminus   \{B(i)\}}\varphi_j(N,w_2,\mathcal{B}) \nonumber \\
    	&\stackrel{{\bf E }}{=}&\sum_{j \in B(i)}\varphi_j(N,w_2,\mathcal{B}).
    \end{eqnarray*}
Similarly, we have $\sum_{j \in B(i)}Ow_j(N,\alpha u_{T}^{N},\mathcal{B})=
    	\sum_{j\in B(i)}Ow_j(N,w_2,\mathcal{B})$. Thus
    \begin{equation}\label{equation12}
\sum_{j \in B(i)}\varphi_j(N,\alpha u_{T}^{N},\mathcal{B})=\sum_{j \in B(i)}Ow_j(N,\alpha u_{T}^{N},\mathcal{B}).
\end{equation}
Note that all players in $T$ are mutually dependent in $\alpha u_{T}^{N}$. Then, by (\ref{equation12}), {\bf SWU} and {\bf N},  we have
    \begin{eqnarray*}
   {|B(i) \cap T|}\varphi_i(N,\alpha u_{T}^{N},\mathcal{B})&\stackrel{{\bf SWU},{\bf N}} {=}&\sum_{j \in B(i)}\varphi_j(N,\alpha u_{T}^{N},\mathcal{B})\stackrel{(\ref{equation12})} {=}\sum_{j \in B(i)}Ow_j(N,w_2,\mathcal{B})\\
   &=&|B(i) \cap T|Ow_i(N,\alpha u_{T}^{N},\mathcal{B}), \ \ \mbox{for $i\in B(i)\cap T$}.
    \end{eqnarray*}
    Therefore, combining with (\ref{equation9}), we conclude that $\varphi(N,\alpha u_{T}^{N},\mathcal{B})=Ow(N,\alpha u_{T}^{N},\mathcal{B})$.
\end{proof}


We turn now to formulating the Owen value axiomatically by replacing
the axioms of symmetry between unions ({\bf SBU}) and marginality ({\bf M}) in the axiomatization \citep{2007_kh} with
axioms of weak mutually dependent between unions ({\bf MBU$^-$}) and super inter-unions marginality ({\bf IUM$^+$}).

\begin{theorem}\label{thm2}
    The Owen value is the unique value on $\mathcal{CG}^N$  that satisfies efficiency {\rm({\bf E})}, symmetry within unions{\rm({\bf SWU})}, weak mutually dependent between unions {\rm({\bf MBU$^-$})} and super inter-unions marginality {\rm({\bf IUM$^+$})}.
\end{theorem}

\begin{proof}
    {\em Existence}.  We have already seen that the Owen value satisfies {\bf E}, {\bf SWU} and {\bf MBU$^-$} in Theorem \ref{thm1}. The formula (\ref{equation4}) implies that the Owen value satisfies  {\bf IUM$^+$}.

    {\em Uniqueness}.   Let $\varphi$ be a value on $\mathcal{CG}^N$ that obeys the above axioms. Set $I(v)=\{T \subseteq N \mid \lambda_T(v) \neq 0\}$. The proof of uniqueness is by induction on $|I(v)|$.

    {\em Induction basis}: 	If $|I(v)|=0$, then all players in $N$ are null and mutually dependent in $v$. By {\bf E}, {\bf SWU} and {\bf MBU$^-$}, $\varphi_i(N,v,\mathcal{B})=Ow_i(N,v,\mathcal{B})=0$ for all $i \in N$.

    {\em Induction hypothesis }({\bf IH}): Assume that $\varphi(N,v,\mathcal{B})=Ow(N,v,\mathcal{B})$ for any  $(N,v,\mathcal{B}) \in {\mathcal{CG}}^N$ with $|I(v)| \le t$.

    {\em Induction step}: We show that $\varphi(N,v,\mathcal{B})=Ow(N,v,\mathcal{B})$ for all $(N,v,\mathcal{B}) \in {\mathcal{CG}}^N$ with $|I(v)| = t+1$. Let $\delta(v)=\{i \in N \mid \forall \ T \in I(v), i \in T\}$, and $\mathcal{B}_{\delta}=\{B \in \mathcal{B}:\ B \cap \delta(v) \neq \emptyset \}$. By (\ref{equation1}) and Observation \ref{ob1} (ii), we see that each player in $\delta(v)$ is  necessary  in $v$.

    We distinguish the following two cases.

    \textbf{Case 1.} For any player $i \in N \setminus   \delta(v)$, without loss of generality, suppose  $i \in B_p \in \mathcal{B}$. Then there exists at least one set $T \in I(v)$ such that $i \notin T$. Define a game \[w=v-\lambda_{T}(v)u^N_{T}.\] Note that $i$ is a null player in  $v-w$, so  $v(Q(R) \cup S\cup\{i\}) - v(Q(R) \cup S) = w(Q(R) \cup S\cup\{i\}) - w(Q(R) \cup S)$ for any $R \subseteq M \setminus   \{p\} $ and $S \subseteq B_p \setminus   \{i\}$, then, by {\bf IUM$^+$} and {\bf IH}, we have
    \begin{equation*}
    	\varphi_i(N,v,\mathcal{B})\stackrel{{\bf IUM^+}}{=}\varphi_i(N,w,\mathcal{B})\stackrel{{\bf IH}}{=}Ow_i(N,w,\mathcal{B})\stackrel{{\bf IUM^+}}{=}Ow_i(N,v,\mathcal{B}).
    \end{equation*}
    Hence,
\begin{equation}\label{equation14}
\varphi_i(N,v,\mathcal{B})=Ow_i(N,v,\mathcal{B})\ \  \text{for all } i \in N \setminus   \delta(v).
\end{equation}
    \textbf{Case 2.} For any player $i \in \delta(v)$, define two games
     \[v'=\sum_{T \in I(v)}\lambda_T(v)u^N_{T'},\ \ \ v''=\sum_{T \in I(v)}\lambda_T(v)u^N_{T''},\] where
     \[T'=T \cup \Big(\bigcup_{B \in \mathcal{B}_{\delta} \setminus \{B(i)\}}B\Big), \ \ \ T''=T \cup \Big(\bigcup_{B \in \mathcal{B}_{\delta}}B\Big).
     \]

 It is easy to see that each player $j \in \bigcup_{B\in \mathcal{B}\setminus \mathcal{B}_{\delta}}B\cup \{B(i)\}$ has the same inter-unions marginal contribution in both $v$ and $v'$, while all players $j \in B \in \mathcal{B} \setminus \{B(i)\}$ have the same inter-unions marginal contribution in both $v'$ and $v''$. Then, by {\bf IUM$^+$}, we have that
    \begin{eqnarray}\label{equation15}
    &&\varphi_j(N,v,\mathcal{B}) = \varphi_j(N,v',\mathcal{B}),\\ \nonumber
    &&Ow_j(N,v,\mathcal{B}) = Ow_j(N,v',\mathcal{B}) \text{, for all } j \in \bigcup_{B\in \mathcal{B}\setminus \mathcal{B}_{\delta}}B\cup \{B(i)\},
    \end{eqnarray}
    and
     \begin{eqnarray}\label{equation16}
    &&\varphi_j(N,v',\mathcal{B}) = \varphi_j(N,v'',\mathcal{B}), \\ \nonumber
    &&Ow_j(N,v',\mathcal{B}) =Ow_j(N,v'',\mathcal{B}) \text{, for all } j \in B \in \mathcal{B} \setminus \{B(i)\}.
    \end{eqnarray}

    For the game $v''$, we show that
    \begin{equation}\label{eq14}
    \varphi_j(N,v'',\mathcal{B})=Ow_j(N,v'',\mathcal{B}) \ \ \mbox{for all $j\in N$}.
    \end{equation}

   For any $ j \in \bigcup_{B\in \mathcal{B}\setminus \mathcal{B}_{\delta}}B$, by  (\ref{equation14}), (\ref{equation15}), and (\ref{equation16}), we have that
    \begin{equation*}
    	\varphi_j(N,v'',\mathcal{B})\stackrel{(\ref{equation16})}{=}\varphi_j(N,v',\mathcal{B})
    \stackrel{(\ref{equation15})}{=}\varphi_j(N,v,\mathcal{B})
    \stackrel{{(\ref{equation14})}}{=}Ow_j(N,v,\mathcal{B}).
    \end{equation*}
 Similarly, we have $Ow_j(N,v'',\mathcal{B})=Ow_j(N,v,\mathcal{B})$ for all $ j \in \bigcup_{B\in \mathcal{B}\setminus \mathcal{B}_{\delta}}B$. Thus $\varphi_j(N,v'',\mathcal{B})=Ow_j(N,v'',\mathcal{B})$ for all $ j \in \bigcup_{B\in \mathcal{B}\setminus \mathcal{B}_{\delta}}B$.
 We next consider the players $j \in \bigcup_{B\in \mathcal{B}_{\delta}}B$.
 Note that all the unions in $\mathcal{B}_{\delta}$ are highly mutually dependent in $v''$ and all players in every $B\in \mathcal{B}_{\delta}$ are mutually dependent in $v''$.
 Then, by {\bf E}, {\bf SWU}, and {\bf MBU$^-$},  we have
    \begin{align*}\label{equation18}
   |B_{\delta}||B(j)|\varphi_j(N,v'',\mathcal{B})&\stackrel{\bf SWU,\, MBU^-}{=}\sum_{j \in \bigcup_{B\in  \mathcal{B}_{\delta}}B}\varphi_j(N,v'',\mathcal{B})\\
   &\stackrel{\bf E}{=}v''(N)-\sum_{k \in \bigcup_{B\in \mathcal{B}\setminus \mathcal{B}_{\delta}}B}\varphi_k(N,v'',\mathcal{B})\\
    &\stackrel{(\ref{eq14})}{=}v''(N)-\sum_{k \in \bigcup_{B\in \mathcal{B}\setminus \mathcal{B}_{\delta}}B}Ow_k(N,v'',\mathcal{B})\\
    &\stackrel{\bf E}{=} \sum_{j \in \bigcup_{B\in  \mathcal{B}_{\delta}}B}Ow_j(N,v'',\mathcal{B})\\
   &\stackrel{\bf SWU,\, MBU^-}{=}|B_{\delta}||B(j)|Ow_j(N,v'',\mathcal{B}).
    \end{align*}
 Hence $\varphi_j(N,v'',\mathcal{B})=Ow_j(N,v'',\mathcal{B})$ for $j \in \bigcup_{B\in \mathcal{B}_{\delta}}B$, and (\ref{eq14}) follows.

   For any player $j \in B \in \mathcal{B} \setminus \{B(i)\}$,  by  (\ref{equation16}),   we have
    \begin{equation}\label{eq15}
    	\varphi_j(N,v',\mathcal{B})\stackrel{(\ref{equation16})}{=}\varphi_j(N,v'',\mathcal{B})\stackrel{(\ref{eq14})} {=}Ow_j(N,v'',\mathcal{B})\stackrel{(\ref{equation16})}{=}Ow_j(N,v',\mathcal{B}).
    \end{equation}

    For any player $j \in B(i) \setminus   \delta(v)$, by (\ref{equation14}), and (\ref{equation15}), we obtain
    \begin{equation}\label{eq16}
    	\varphi_j(N,v',\mathcal{B})\stackrel{(\ref{equation15})}{=}\varphi_j(N,v,\mathcal{B})
    \stackrel{{(\ref{equation14})}}{=}Ow_j(N,v,\mathcal{B})\stackrel{(\ref{equation15})}{=}Ow_j(N,v',\mathcal{B}).
    \end{equation}
    Note that all players in $B(i) \cap \delta(v)$ are mutually dependent in $v'$. Then,  by {\bf E}, {\bf SWU},
    (\ref{equation15}), (\ref{eq15}), and (\ref{eq16}), we obtain
    \begin{eqnarray*}
    	|B(i) \cap \delta(v)|\varphi_i(N,v,\mathcal{B}) &\overset{(\ref{equation15})}{=}& |B(i) \cap \delta(v)|\varphi_i(N,v',\mathcal{B})\\
    	&\stackrel{{\bf E,\, SWU}}{=}&
    	v'(N)-\sum_{j \in N \setminus   (B(i) \cap \delta(v))}\varphi_j(N,v',\mathcal{B}) \\
   &\stackrel{(\ref{eq15}),\, (\ref{eq16})}{=}& v'(N)-\sum_{j \in N \setminus   (B(i) \cap \delta(v))}Ow_j(N,v',\mathcal{B}) \\
   &\stackrel{{\bf E,\, SWU}}{=}& |B(i) \cap \delta(v)|Ow_i(N,v',\mathcal{B})\\
   &\overset{(\ref{equation15})}{=}&|B(i) \cap \delta(v)|Ow_i(N,v,\mathcal{B}).
    \end{eqnarray*}
    Hence, $\varphi_i(N,v,\mathcal{B})=Ow_i(N,v,\mathcal{B})$  for any  player $i \in B(i) \cap \delta(v)$. Thus $\varphi_i(N,v,\mathcal{B})=Ow_i(N,v,\mathcal{B})$  for all  player $i \in \delta(v)$.
Therefore, combining with (\ref{equation14}),  we conclude that $\varphi(N,v,\mathcal{B})=Ow(N,v,\mathcal{B})$ for all $i\in N$.
\end{proof}


Finally, we show that the axiom of differential marginality of inter-mutually dependent unions ({\bf DMU$_{md}^-$}), jointly with the axioms
of efficiency ({\bf E}),  differential marginality of mutually dependent players within unions ({\bf UDM$_{md}$}) and the null player out property ({\bf NPO}), characterizes the Owen value for variable
 player sets.


\begin{theorem}\label{theorem3}
    The Owen value is the unique value on $\mathcal{CG}$ that satisfies efficiency {\rm({\bf E})},  differential marginality of mutually dependent players within unions {\rm({\bf UDM$_{md}$})}, differential marginality of inter-mutually dependent unions {\rm({\bf DMU$_{md}^-$})} and the null player out property {\rm({\bf NPO})}.
\end{theorem}

\begin{proof}
    {\em Existence}.   According to \cite{2024_he}, it is straightforward to know that the Owen value satisfies {\bf E}, {\bf NPO}, {\bf UDM$_{md}$} and {\bf DMU$_{md}^-$}.

    {\em Uniqueness}.   Let $\varphi$ be a value on $\mathcal{CG}$ that obeys the above axioms. We have to show that $\varphi=Ow$.  Let $I(v)=\{T \subseteq N \mid \lambda_T(v) \neq \emptyset\}$. We prove the assertion by induction on $|I(v)|$.

    It is easy to see  that {\bf NPO} and {\bf E} imply {\bf N}. So  $\varphi$  satisfies {\bf N}.

    {\em Induction basis}: If $|I(v)|=0$, then each $i \in N$ is a null player in $v$. So, by {\bf N}, $\varphi_i(N,v,\mathcal{B})=Ow_i(N,v,\mathcal{B})=0$ for all $i \in N$.

    {\em Induction hypothesis} ({\bf IH}): Assume that $\varphi(N,v,\mathcal{B})=Ow(N,v,\mathcal{B})$ for any $(N,v,\mathcal{B}) \in \mathcal{CG}$ with $|I(v)|\le t$.

    {\em Induction step} :  We show that $\varphi(N,v,\mathcal{B})=Ow(N,v,\mathcal{B})$ for any $(N,v,\mathcal{B}) \in \mathcal{CG}$ with $|I(v)| = t+1$. Let $\delta(v) = \{i \in N \mid \forall \, T \in I(v), i \in T\}$, $\beta(v) = \{i \in N \mid \forall\, T \in I(v), i \notin T\}$. Then $\delta(v)\cap \beta(v)=\emptyset$. Obviously, each player in $\delta(v)$ is necessary in $v$, while each player in $\beta(v)$ is null in $v$.

     For any player $i \in N \setminus   \delta(v)$,  there exists one set $T \in I(v) $ such that $T \not\ni i$. Take a player $l' \in \mathcal{U} \setminus   N$ arbitrarily. We construct the  coalition structure $\mathcal{B}'= \mathcal{B} \setminus \{B(i)\} \cup \{B(i) \cup \{l'\}\}$,
and define
    \[w = \sum_{T \in I(v)}\lambda_T(v)u_{T}^{N \cup \{l'\}}, \ \ w' = \sum_{T \in I(v) \setminus   \{T\}}\lambda_T(v)u_{T}^{N \cup \{l'\}},\ \ v' = \sum_{T \in I(v) \setminus   \{T\}}\lambda_T(v)u_{T}^{N}.\]
    Then $I(v')<I(v)$. Note that $i$ and $l'$ are mutually dependent in $w-w'$, and $l'$ is a null player in $w$ and $w'$. By the induction hypothesis, {\bf NPO}, {\bf N}, {\bf UDM$_{md}$} and (\ref{equation3}), we have
    \begin{align*}
    	\varphi_i(N,v,\mathcal{B})
    	&\stackrel{\bf NPO}{=} \varphi_i(N \cup \{l'\},w,{\mathcal{B}}^{'}) \nonumber \\
    	&\stackrel{\bf N}{=} \varphi_i(N \cup \{l'\},w,{\mathcal{B}}^{'}) - \varphi_{l'}(N \cup \{l'\},w,{\mathcal{B}}^{'}) \nonumber \\
    	&\stackrel{\bf UDM_{md}}{=} \varphi_i(N \cup \{l'\},w',{\mathcal{B}}^{'}) - \varphi_{l'}(N \cup \{l'\},w',{\mathcal{B}}^{'}) \nonumber \\
    	&\stackrel{\bf N}{=} \varphi_i(N \cup \{l'\},w',{\mathcal{B}}^{'})\stackrel{\bf NPO}{=}\varphi_i(N,v',\mathcal{B}) \nonumber \\
    	&\stackrel{{\bf IH}}{=} Ow_i(N,v', \mathcal{B}) \stackrel{(\ref{equation3})}{=} Ow_i(N,v, \mathcal{B}).
    \end{align*}
Hence, \begin{align}\label{eq17}
    	\varphi_i(N,v,\mathcal{B})=Ow_i(N,v, \mathcal{B}),  \ \ \mbox{for all $i\in N\setminus \delta(v)$}.
\end{align}

It remains to show that $\varphi_i(N,v,\mathcal{B})=Ow_i(N,v, \mathcal{B})$  for all $i\in  \delta(v)$.
Let $i\in \delta(v)$.
    We distinguish the following two cases.

    \textbf{Case 1.}     $B(i) \setminus   (\delta(v) \cup \beta(v))\neq \emptyset$.
Let \(j \in B(i) \setminus (\delta(v) \cup \beta(v))\). Then there exists  \(T \in I(v)\) such that \(j \in T\). Define
    \[v'' = v - \lambda_{T}(v)u^N_{T}.\]
 Note that $I(v'')<I(v)$, and  $i,j\in T$  are mutually dependent in $v-v''$. Thus, combining {\bf UDM$_{md}$} and (\ref{eq17}) with {\bf A} and {\bf SWU} of the Owen value, we have
    \begin{align*}\label{equation23}
    	\varphi_i(N,v,\mathcal{B})-Ow_i(N,v'',\mathcal{B})
    	&\stackrel{\bf IH}{=}
    	\varphi_i(N,v,\mathcal{B})-\varphi_i(N,v'',\mathcal{B}) \nonumber \\
    	&\stackrel{\bf UDM_{md}}{=}
    	\varphi_j(N,v,\mathcal{B})-\varphi_j(N,v'',\mathcal{B}) \nonumber \\
    	&\stackrel{ (\ref{eq17}),\, {\bf IH}}{=}Ow_j(N,v, \mathcal{B})-Ow_j(N,v'', \mathcal{B})\nonumber \\
    	&\stackrel{\bf A}{=}Ow_j(N,v-v'', \mathcal{B})\stackrel{\bf SWU}{=}Ow_i(N,v-v'', \mathcal{B})  \nonumber \\
    	&\stackrel{\bf A}{=}	Ow_i(N,v,\mathcal{B})-Ow_i(N,v'',\mathcal{B}).
    \end{align*}
    Therefore, $\varphi_i(N,v,\mathcal{B})=Ow_i(N,v,\mathcal{B})$.

    \textbf{Case 2.}  \( B(i) \setminus   (\delta(v) \cup \beta(v)) = \emptyset \).
  Let
    $\mathcal{B}'= \{ B \in \mathcal{B}_{\delta} \mid B \setminus   (\delta(v) \cup \beta(v)) = \emptyset \}$ where $\mathcal{B}_{\delta}=\{B \in \mathcal{B}:\ B \cap \delta(v) \neq \emptyset \}$. Then $B(i)\in \mathcal{B}'$, and for each $B\in \mathcal{B}'$, we have
    $B=(B\cap \delta(v))\cup(B\cap \beta(v))$.

    For any unions \( B_{p}, B_{q} \in \mathcal{B}' \), it is easy to see that $B_{p}\cap \delta(v)$ and $B_{q}\cap \delta(v)$ in $\mathcal{B}_{|\delta(v)}$ are highly mutually dependent in \( v \). By {\bf NPO}, {\bf N}, and {\bf DMU$_{md}^-$}, we have
\begin{eqnarray*}
    	&&\sum_{j \in B_{p}}\varphi_j (N,v,\mathcal{B}) - \sum_{j \in B_{q}}\varphi_j (N,v,\mathcal{B})\\
    	&\,&\stackrel{\bf N}{=}
    	\sum_{j \in B_{p}\cap\delta(v)}\varphi_j (N,v,\mathcal{B}) - \sum_{j \in B_{q}\cap\delta(v)}\varphi_j (N,v,\mathcal{B}) \nonumber \\
    	&\,&\stackrel{\bf NPO}{=}
    	\sum_{j \in B_{p}\cap\delta(v)}\varphi_j (N \setminus   \beta(v),v,\mathcal{B}_{|N \setminus   \beta(v)}) - \sum_{j \in B_{q}\cap\delta(v)}\varphi_j (N \setminus   \beta(v),v,\mathcal{B}_{|N \setminus   \beta(v)}) \nonumber \\
    	&\,&\stackrel{\bf DMU_{md}^-}{=}
    	\sum_{j \in B_{p}\cap\delta(v)}\varphi_j (N \setminus   \beta(v),{\bf 0},\mathcal{B}_{|N \setminus   \beta(v)}) - \sum_{j \in B_{q}\cap\delta(v)}\varphi_j (N \setminus   \beta(v),{\bf 0},\mathcal{B}_{|N \setminus   \beta(v)}) \nonumber \\
    	&\,&\stackrel{\bf N}{=} 0
    \end{eqnarray*}
Hence,
\begin{eqnarray}\label{eq18}
\sum_{j \in B_{p}\cap\delta(v)}\varphi_j (N,v,\mathcal{B})=\sum_{j \in B_{q}\cap\delta(v)}\varphi_j (N,v,\mathcal{B})
 \ \ \mbox{for all  $B_{p}, B_{q} \in \mathcal{B}'$}.
\end{eqnarray}
Since \( j, k \in B \cap \delta(v) \) for any $B\in \mathcal{B}'$ are necessary players in $v$, they are  mutually dependent in $v$. Thus, by {\bf UDM$_{md}$} and {\bf N}, we have
    \begin{equation*}
    	\varphi_j (N,v,\mathcal{B}) - \varphi_k (N,v,\mathcal{B})
    	\stackrel{\bf UDM_{md}}{=}
    	\varphi_j (N,{\bf 0}, \mathcal{B}) - \varphi_k (N,{\bf 0}, \mathcal{B}) \stackrel{\bf N}{=} 0.
    \end{equation*}
Therefore,
\begin{eqnarray}\label{eq19}
\varphi_j (N,v,\mathcal{B})=\varphi_k (N,v,\mathcal{B})\ \  \mbox{for all $B\in \mathcal{B}'$ and $j,k\in B\cap \delta(v)$}.
\end{eqnarray}
Similarly, we have $\sum_{j \in B_{p}\cap\delta(v)}Ow_j (N,v,\mathcal{B})=\sum_{j \in B_{q}\cap\delta(v)}Ow_j (N,v,\mathcal{B})$ for all  \( B_{p}, B_{q} \in \mathcal{B}' \), and $Ow_j (N,v,\mathcal{B})=Ow_k (N,v,\mathcal{B})$ for all $B\in \mathcal{B}'$ and $j,k\in B\cap \delta(v)$.
Combined with {\bf Case 1},  {\bf E}, and (\ref{eq17})-(\ref{eq19}), we have
\begin{align*}
|\mathcal{B}'|B(i)|\cap \delta(v)|\varphi_i (N,v,\mathcal{B})=|\mathcal{B}'|B(i)|\cap \delta(v)|Ow_i (N,v,\mathcal{B}).
\end{align*}
Consequently, $\varphi_i (N,v,\mathcal{B})=Ow_i (N,v,\mathcal{B})$.

By (\ref{eq17}), {\bf Case 1} and {\bf Case 2}, we conclude that $\varphi_i (N,v,\mathcal{B})=Ow_i (N,v,\mathcal{B})$ for all $i\in N$.
\end{proof}

\subsection{The logical independence of the axioms in axiomatizations}
In this subsection we show the logical independence of the axioms in our axiomatizations.

\begin{remark}
{\rm
The axioms involved in Theorem \ref{thm1} are logically independent.
\begin{itemize}
	
	\item The null allocation rule $\varphi^0={\bf 0}$ on $\mathcal{CG}$ satisfies {\bf N}, {\bf A}, {\bf SWU}, {\bf IAG} and {\bf MBU$^-$}, but not {\bf E}.
	
	\item The $SE$ value \citep{2023_abe} on $\mathcal{CG}$, defined by
	\begin{eqnarray*}
		SE_i(N,v,\mathcal{B})=\sum_{\emptyset \neq T \subseteq N : B(i) \cap T \neq \emptyset} \frac{\lambda_T(v)}{|B(i)|m_T},
	\end{eqnarray*}
	satisfies {\bf E}, {\bf A}, {\bf SWU}, {\bf IAG} and {\bf MBU$^-$}, but not {\bf N}.
	
	\item The allocation rule $\varphi^1$ on $\mathcal{CG}$, defined by
	\[
	\varphi^1_i(N,v,\mathcal{B}) =
	\begin{cases}
		0, & |\mathcal{B}| = |N|, \, i \in S, \\
		\frac{v(N)}{|N \setminus S|}, & |\mathcal{B}| = |N|, \, i \in N \setminus   S, \\
		Ow_i(N,v,\mathcal{B}), & \text{otherwise},
	\end{cases}
	\]
	where $S=\{i \in N \mid i \text{ is a null player in } v \}$, satisfies {\bf E}, {\bf N}, {\bf SWU}, {\bf IAG} and {\bf MBU$^-$}, but not {\bf A}.
	
	\item The allocation rule $\varphi^2$ on $\mathcal{CG}$, defined by
	\[\varphi^2_i(N,v,\mathcal{B})=\sum_{\emptyset \neq T \subseteq N : i \in T} \frac{w_i\lambda_T(v)}{\sum_{j \in T \cap B(i)}w_j m_T},\]
	where $w_i$ is the given exogenous weight for player $i$, satisfies {\bf E}, {\bf N}, {\bf A}, {\bf IAG} and {\bf MBU$^-$}, but not {\bf SWU}.
	
	\item The allocation rule $\varphi^3$ on $\mathcal{CG}$, defined by
	\[
	\varphi^3_i(N,v,\mathcal{B}) =
	\begin{cases}
		\sum_{ T \subseteq N :i \in T}\frac{\lambda_T(v)|B(i)|}{\sum_{j \in T}|B(j)|}, & \mathcal{B} = \mathcal{B} ', \\
		Ow_i(N,v,\mathcal{B}), & \text{otherwise},
	\end{cases}
	\]
	where for any $B_1,B_2 \in \mathcal{B}'$, $|B_1|=|B_2|$, satisfies
	{\bf E}, {\bf N}, {\bf A}, {\bf SWU} and {\bf MBU$^-$}, but not {\bf IAG}.
	
	\item The $Ow^P$ value on $\mathcal{CG}$, defined by
	\[
	Ow_i^{P}(N, v, \mathcal{B}) = \sum_{\emptyset \neq T \subseteq N :  T \ni i} \frac{\lambda_T(v)|B(i)|}{|B(i)\cap T|\sum_{B\in \mathcal{B}:B \cap T \neq \emptyset} |B|},
	\]
	satisfies {\bf E}, {\bf N}, {\bf A}, {\bf SWU} and {\bf IAG}, but not {\bf MBU$^-$}.
	
\end{itemize}}
\end{remark}

\begin{remark}
{\rm
The axioms involved in Theorem \ref{thm2} are logically independent.

\begin{itemize}

    \item The null allocation rule $\varphi^0={\bf 0}$ on $\mathcal{CG}$ satisfies {\bf SWU}, {\bf MBU$^-$} and {\bf IUM$^+$}, but not {\bf E}.

    \item The allocation rule $\varphi^2$ on $\mathcal{CG}$ satisfies {\bf E}, {\bf MBU$^-$} and {\bf IUM$^+$}, but not {\bf SWU}.

    \item The $Ow^{P}$ value on $\mathcal{CG}$ satisfies {\bf E}, {\bf SWU} and {\bf IUM$^+$}, but not {\bf MBU$^-$}.

    \item The SE value on $\mathcal{CG}$ satisfies {\bf E} , {\bf SWU} and {\bf MBU$^-$}, but not {\bf IUM$^+$}.

\end{itemize}}
\end{remark}

\begin{remark}
{\rm
The axioms involved in Theorem \ref{theorem3} are logically independent.

\begin{itemize}
    \item The null allocation rule $\varphi^0={\bf 0}$ on $\mathcal{CG}$ satisfies {\bf NPO} , {\bf UDM$_{md}$} and {\bf DMU$_{md}^-$}, but not {\bf E}.

    \item The allocation rule $\varphi^4$ on $\mathcal{CG}$, defined by
    \[
    \varphi^4_i(N,v,\mathcal{B}) =
    Ow_i(N,v,\mathcal{B})-v(\{i\}) + \frac{\sum_{j \in B(i)}v(\{j\})}{|B(i)|},
    \]
    satisfies {\bf E}, {\bf UDM$_{md}$} and {\bf DMU$_{md}^-$}, but not {\bf NPO}.

    \item  For every $(N,v,B)\in \mathcal{CG}$, setting $B'\in \mathcal{B}$, the value $\varphi^5$ defined by
    \[
    \varphi_i^5(N,v,B)=
    \begin{cases}
   	 Ow_i (N,v,B) - v(\{i\}) + \frac{\sum_{j\in B'} v(\{j\})}{|B'|}, & \text{if } \lambda_N(v)\neq 0, i\in B', \\
   	 Ow_i (N,v,B), & \text{otherwise},
   \end{cases}
    \]
    satisfies {\bf E}, {\bf NPO} and {\bf DMU$_{md}^-$}, but not {\bf UDM$_{md}$}.

    \item The Shapley value on $\mathcal{CG}$ satisfies {\bf E}, {\bf NPO} and {\bf UDM$_{md}$}, but not {\bf DMU$_{md}^-$}.

\end{itemize}}
\end{remark}

\section{Conclusions}\label{section4}


In this paper we introduce the concepts of highly mutually dependent unions,  the inter-unions marginal contributions  and union-wise mutually dependent games.
Based on these notions, we propose several axioms: weak mutually dependent between unions, differential marginality of inter-mutually dependent unions, super inter-unions marginality and invariance across games. We establish three axiomatic characterizations of the Owen value by combining the above axioms with several standard axioms.

It is worth noting that, the requirement for the axiom of weak mutually dependent between unions that two unions  are highly mutually dependent in $v$ can not be replaced by two unions are highly symmetric (any pair of players, with one from each of the two unions, are symmetric
in the game). This reason is that two unions are highly symmetric does not imply that they are symmetric.
This implies that the Owen value does not satisfy
the modified property. For future study, just as \cite{1989_winter} generalized the Owen value to TU-games with level structure, it remains an open question to generalize our characterizations of the Owen value to TU-games with level structure.


\end{document}